\documentclass[11pt]{article}
\usepackage{fullpage}

\usepackage{times}
\usepackage{epsfig}
\usepackage{amsmath}
\usepackage{amssymb}
\usepackage{amstext}
\usepackage{amsmath}
\usepackage{xspace}
\usepackage{theorem}

\newtheorem{theorem}{Theorem}
\newtheorem{definition}{Definition}

\newtheorem{lemma}[theorem]{Lemma}

\newcommand{\qed}{\mbox{\ \ \ }\rule{6pt}{7pt} \bigskip}

\newcommand{\comment}[1]{}
\newenvironment{proof}{\noindent{\em Proof:}}{\hfill\qed}

\newenvironment{proofof}[1]{\noindent{\em Proof of #1:}}{\hfill\qed}

\newcommand{\Path}{{\mathcal P}}

\newcommand{\C}{{\mathcal C}}
\newcommand{\A}{{\mathcal A}}

\newcommand{\M}{{\mathcal M}}
\newcommand{\E}{{\mathcal E}}
\newcommand{\G}{{\mathcal G}}
\newcommand{\ld}[1]{{\ell_e^{#1}}}
\newcommand{\lef}{{\ld{\C}}}
\newcommand{\lei}{{\ld{\A}}}
\newcommand{\OPT}{{\textrm {OPT}}}

\title{Packing multiway cuts in capacitated graphs\footnote{The conference version of this paper is to appear at SODA 2009. This is the full version.}}
\author{Siddharth Barman\thanks{Computer Sciences Dept., University of
    Wisconsin - Madison, {\tt sid@cs.wisc.edu}. Supported in part by
    NSF award CCF-0643763.} \and Shuchi Chawla\thanks{Computer
    Sciences Dept., University of Wisconsin - Madison, {\tt
      shuchi@cs.wisc.edu}. Supported in part by NSF awards
    CCF-0643763 and CCF-0830494.}}  

\date{\today}

\begin{document}
\maketitle{}

\thispagestyle{empty}

\begin{abstract}
We consider the following ``multiway cut packing" problem in
undirected graphs: we are given a graph $G=(V,E)$ and $k$ commodities,
each corresponding to a set of terminals located at different vertices
in the graph; our goal is to produce a collection of cuts
$\{E_1,\cdots,E_k\}$ such that $E_i$ is a multiway cut for commodity
$i$ and the maximum load on any edge is minimized. The load on an edge
is defined to be the number of cuts in the solution crossing the
edge. In the capacitated version of the problem edges have capacities
$c_e$ and the goal is to minimize the maximum {\em relative} load on
any edge -- the ratio of the edge's load to its capacity. We present
the first constant factor approximations for this problem in arbitrary
undirected graphs. The multiway cut packing problem arises in the
context of graph labeling problems where we are given a partial
labeling of a set of items and a neighborhood structure over them,
and, informally stated, the goal is to complete the labeling in the
most consistent way. This problem was introduced by Rabani, Schulman,
and Swamy (SODA'08), who developed an $O(\log n/\log \log n)$
approximation for it in general graphs, as well as an improved
$O(\log^2 k)$ approximation in trees. Here $n$ is the number of nodes
in the graph.

We present an LP-based algorithm for the multiway cut packing problem
in general graphs that guarantees a maximum edge load of at most
$8\OPT+4$. Our rounding approach is based on the observation that
every instance of the problem admits a laminar solution (that is, no
pair of cuts in the solution crosses) that is near-optimal. For the
special case where each commodity has only two terminals and all
commodities share a common sink (the ``common sink $s$-$t$ cut
packing" problem) we guarantee a maximum load of $\OPT+2$. Both of
these variants are NP-hard; for the common-sink case our result is
nearly optimal.
\end{abstract}

\newpage
\setcounter{page}{1}

\section{Introduction}

We study the {\em multiway cut packing} problem (MCP) introduced by
Rabani, Schulman and Swamy~\cite{RSS08}. In this problem, we are given
$k$ instances of the multiway cut problem in a common graph, each
instance being a set of terminals at different locations in the graph.
Informally, our goal is to compute nearly-disjoint multiway cuts for
each of the instances. More precisely, we aim to minimize the maximum
number of cuts that any single edge in the graph belongs to. In the
weighted version of this problem, different edges have different
capacities; the goal is to minimize the maximum relative load of any
edge, where the relative load of an edge is the ratio of the number of
cuts it belongs to and its capacity.

The multiway cut packing problem belongs to the following class of
graph labeling problems. We are given a partially labeled set of $n$
items along with a weighted graph over them that encodes similarity
information among them. An item's label is a string of length $k$
where each coordinate of the string is either drawn from an alphabet
$\Sigma$, or is undetermined. Roughly speaking, the goal is to
complete the partial labeling in the most consistent possible
way. Note that completing a single specific entry (coordinate) of each
item label is like finding what we call a ``set multiway cut''---for
$\sigma\in\Sigma$ let $S^i_{\sigma}$ denote the set of nodes for which
the $i$th coordinate is labeled $\sigma$ in the partial labeling, then
a complete and consistent labeling for this coordinate is a partition
of the items into $|\Sigma|$ parts such that the $\sigma^{\text{th}}$
part contains the entire set $S^i_{\sigma}$. The cost of the labeling
for a single pair of neighboring items in the graph is measured by the
Hamming distance between the labels assigned to them. The overall cost
of the labeling can then be formalized as a certain norm of the vector
of (weighted) edge costs.

Different choices of norms for the overall cost give rise to different
objectives. Minimizing the $\ell_1$ norm, for example, is the same as
minimizing the sum of the edge costs. This problem decomposes into
finding $k$ minimum set multiway cuts. Each set multiway cut instance
can be reduced to a minimum multiway cut instance by simply merging
all the items in the same set $S_{\sigma}$ into a single node in the
graph, and can therefore be approximated to within a factor of
$1.5$~\cite{CKR00}. On the other hand, minimizing the $\ell_{\infty}$
norm of edge costs (equivalently, the maximum edge cost) becomes the
set multiway cut packing problem. Formally, in this problem, we are
given $k$ set multiway cut instances $S^1, \cdots, S^k$, where each
$S^i=S^i_1\times S^i_2\times \cdots\times S^i_{|\Sigma|}$. The goal is
to find $k$ cuts, with the $i$th cut separating every pair of
terminals that belong to sets $S^i_{j_1}$ and $S^i_{j_2}$ with $j_1\ne
j_2$, such that the maximum (weighted) cost of any edge is
minimized. When $|S^i_j|=1$ for all $i\in [k]$ and $j\in\Sigma$, this
is the multiway cut packing problem.

To our knowledge Rabani et al.~\cite{RSS08} were the first to consider
the multiway cut packing problem and provide approximation algorithms
for it. They used a linear programming relaxation of the problem along
with randomized rounding to obtain an $O(\frac{\log n}{\log \log n})$
approximation, where $n$ is the number of nodes in the given
graph\footnote{Rabani et al. claim in their paper that the same
approximation ratio holds for the set multiway cut packing problem
that arises in the context of graph labelings. However their approach
of merging nodes with the same attribute values (similar to what we
described above for minimizing the $\ell_1$ norm of edge costs) does
not work in this case. Roughly speaking, if nodes $u$ and $v$ have the
same $i$th attribute, and nodes $v$ and $w$ have the same $j$th
attribute, then this approach merges all three nodes, although an
optimal solution may end up separating $u$ from $w$ in some of the
cuts. We are not aware of any other approximation preserving reduction
between the two problems.}. This approximation ratio arises from an
application of the Chernoff bounds to the randomized rounding process,
and improves to an $O(1)$ factor when the optimal load is $\Omega(\log
n)$. When the underlying graph is a tree, Rabani et al. use a more
careful deterministic rounding technique to obtain an improved
$O(\log^2 k)$ approximation. The latter approximation factor holds
also for a more general multicut packing problem (described in more
detail below). One nice property of the latter approximation is that
it is independent of the size of the graph, and remains small as the
graph grows but $k$ remains fixed. Then, a natural open problem
related to their work is whether a similar approximation guarantee
independent of $n$ can be obtained even for general graphs.

{\bf Our results \& techniques.} We answer this question in the
positive. We employ the same linear programming relaxation for this
problem as Rabani et al., but develop a very different rounding
algorithm. In order to produce a good integral solution our rounding
algorithm requires a fractional collection of cuts that is not only
feasible for the linear program but also satisfies an additional good
property---the cut collection is laminar. In other words, when
interpreted appropriately as subsets of nodes, no two cuts in the
collection ``cross'' each other. Given such an input the rounding
process only incurs a small additive loss in performance---the final
(absolute) load on any edge is at most $3$ more than the load on that
edge of the fractional solution that we started out with. Of course
the laminarity condition comes at a cost -- not every fractional
solution to the cut packing LP can be interpreted as a laminar
collection of cuts (see, e.g., Figure~\ref{fig:lam-gap}). We show that
for the multiway cut problem any fractional collection of cuts can be
converted into a laminar one while losing only a multiplicative factor
of $8$ and an additive $o(1)$ amount in edge loads. Therefore, for
every edge $e$ we obtain a final edge load of $8\ld{\OPT}+4$, where
$\ld{\OPT}$ is the optimal load on the edge. We only load edges with
$c_e\ge 1$ and since the optimal cost is at least $1$ our algorithm
also obtains a purely multiplicative $12$ approximation.

Our laminarity based approach proves even more powerful in the special
case of {\em common-sink $s$-$t$ cut packing} problem or CSCP. In this
special case every multiway cut instance has only two terminals and
all the instances share a common sink $t$. We use these properties to
improve both the rounding and laminarity transformation algorithms,
and ensure a final load of at most $\ld{\OPT}+1$ for every edge
$e$. The CSCP is NP-hard (see Section~\ref{sec:NP-hard}) and so our
guarantee for this special case is the best possible.

In converting a fractional laminar solution to an integral one we use
an iterative rounding approach, assigning an integral cut at each
iteration to an appropriate ``innermost'' terminal. Throughout the
algorithm we maintain a partial integral cut collection and a partial
fractional one and ensure that these collections together are feasible
for the given multiway cut instances. As we round cuts, we ``shift''
or modify other fractional cuts so as to maintain bounds on edge
loads. Maintaining feasibility and edge loads simultaneously turns out
to be relatively straightforward in the case of common-sink $s$-$t$
cut packing -- we only need to ensure that none of the cuts in the
fractional or the integral collection contain the common sink
$t$. However in the general case we must ensure that new fractional
cuts assigned to any terminal must exclude all other terminals of the
same multiway cut instance. This requires a more careful reassignment
of cuts.

{\bf Related work.} Problems falling under the general framework of
graph labeling as described above have been studied in various
guises. The most extensively studied special case, called label
extension, involves partial labelings in which every item is either
completely labeled or not labeled at all. When the objective is to
minimize the $\ell_1$ norm of edge costs, this becomes a special case
of the metric labeling and 0-extension
problems~\cite{KT02,CKR04,CKNZ04,Kar98}. (The main difference between
0-extension and the label extension problem as described above is that
the cost of the labeling in the former arises from an arbitrary metric
over the labels, while in the latter it arises from the Hamming
metric.)

When the underlying graph is a tree and edge costs are given by the
edit distance between the corresponding labels, this is known as the
tree alignment problem. The tree alignment problem has been studied
widely in the computational biology literature and arises in the
context of labeling phylogenies and evolutionary trees. This version
is also NP-hard, and there are several PTASes
known~\cite{WJL96,WJG00,WG97}. Ravi and Kececioglu~\cite{RK98} also
introduced and studied the $\ell_{\infty}$ version of this problem,
calling it the bottleneck tree alignment problem.
They presented an $O(\log n)$ approximation for this problem. A
further special case of the label extension problem under the
$\ell_{\infty}$ objective, where the underlying tree is a star with
labeled leaves, is known as the closest string problem. This problem
is also NP-hard but admits a PTAS~\cite{LMW02}.

As mentioned above, the multiway cut packing problem was introduced by
Rabani, Schulman and Swamy~\cite{RSS08}. Rabani et al. also studied
the more general multicut packing problem (where the goal is to pack
multicuts so as to minimize the maximum edge load) as well as the
label extension problem with the $\ell_{\infty}$ objective. Rabani et
al. developed an $O(\log^2 k)$ approximation for multicut packing in
trees, and an $O(\log M\frac{\log n}{\log\log n})$ in general
graphs. Here $M$ is the maximum number of terminals in any one
multicut instance. For the label extension problem they presented a
constant factor approximation in trees, which holds even when edge
costs are given by a fairly general class of metrics over the label
set (including Hamming distance as well as edit distance).

Another line of research loosely related to the cut packing problems
described here considers the problem of finding the largest collection
of edge-disjoint cuts (not corresponding to any specific terminals) in
a given graph. While this problem can be solved exactly in polynomial
time in directed graphs~\cite{LY78}, it is NP-hard in undirected
graphs, and Caprara, Panconesi and Rizzi~\cite{CPR04} presented a $2$
approximation for it. In terms of approximability, this problem is
very different from the one we study---in the former, the goal is to
find as many cuts as possible, such that the load on any edge is at
most $1$, whereas in our setting, the goal is to find cuts for all the
commodities, so that the maximum edge load is minimized.

\section{Definitions and results}

Given a graph $G=(V,E)$, a {\em cut} in $G$ is a subset of edges $E'$,
the removal of which disconnects the graph into multiple connected
components. A {\em vertex partition} of $G$ is a pair $(C,V\setminus
C)$ with $\emptyset\subsetneq C\subsetneq V$. For a set $C$ with
$\emptyset\subsetneq C\subsetneq V$, we use $\delta(C)$ to denote the
cut defined by $C$, that is, $\delta(C)=\{ (u,v)\in E:
|C\cap\{u,v\}|=1\}$. We say that a cut $E'\subseteq E$ separates
vertices $u$ and $v$ if $u$ and $v$ lie in different connected
components in $(V,E\setminus E')$. The vertex partition defined by set
$C$ separates $u$ and $v$ if the two vertices are separated by the cut
$\delta(C)$. Given a collection of cuts $\E=\{E_1,\cdots,E_k\}$ and
capacities $c_e$ on edges, the load $\ld{\E}$ on an edge $e$ is
defined as the number of cuts that contain $e$, that is,
$\ld{\E}=|\{E_i\in\E | e\in E_i\}|$. Likewise, given a collection of
vertex partitions $\C=\{C_1, \cdots, C_k\}$, the load $\ld{\C}$ on an
edge $e$ is defined to be the load of the cut collection
$\{\delta(C_1), \cdots, \delta(C_k)\}$ on $e$.

The input to a {\em multiway cut packing} problem (MCP) is a graph
$G=(V,E)$ with non-zero integral capacities $c_e$ on edges, and $k$
sets $S_1,\cdots, S_k$ of terminals (called ``commodities''); each
terminal $i\in S_a$ resides at a vertex $r_i$ in $V$. The goal is to
produce a collection of cuts $\E=\{E_1,\cdots, E_k\}$, such that (1)
for all $a\in [k]$, and for all pairs of terminals $i,j\in S_a$, the
cut $E_a$ separates $r_i$ and $r_j$, and (2) the maximum ``relative
load'' on any edge, $\max_e \ld{\E}/c_e$, is minimized.

In a special case of this problem called the {\em common-sink $s$-$t$
cut packing} problem (CSCP), the graph $G$ contains a special node $t$
called the sink and each commodity set has exactly two terminals, one
of which resides at $t$. Again the goal is to produce a collection of
cuts, one for each commodity such that the maximum relative edge load
is minimized.

Both of these problems are NP-hard to solve optimally (see
Section~\ref{sec:NP-hard}), and we present LP-rounding based
approximation algorithms for them. We assume without loss of
generality that the optimal solution has a relative load of $1$. The
integer program~\ref{eqn:IP} below encodes the set of solutions to the
MCP with relative load $1$.

Here $\Path_a$ denotes the set of all paths between any two vertices
$r_i,r_j$ with $i,j\in S_a$, $i\ne j$. In order to be able to solve
this program efficiently, we relax the final constraint to $x_{a,e}\in
[0,1]$ for all $a\in [k]$ and $e\in E$. Although the resulting linear
program has an exponential number of constraints, it can be solved
efficiently; in particular, the polynomial-size program~\ref{eqn:LP}
below is equivalent to it. Given a feasible solution to this linear
program, our algorithms round it into a feasible integral solution
with small load.

\begin{center}
\begin{small}
\begin{tabular}{cc}
\fbox{
\parbox{0.45\textwidth}{
\begin{align*}
\sum_{e\in P} x_{a,e} & \ge 1 & & \forall a\in [k], P\in \Path_a \\
\sum_a x_{a,e} & \le c_e & & \forall e\in E\\
x_{a,e} & \in \{0,1\} & & \forall a\in [k], e\in E\\
\tag{\bf{MCP-IP}} \label{eqn:IP}
\end{align*}
}}
&
\fbox{
\parbox{0.45\textwidth}{
\begin{align*}
d_a(u,v) & \le d_a(u,w) + d_a(w,v) & & \forall a\in [k], u,v,w\in V \\
d_a(r_i,r_j) & \ge 1 & & \forall a\in [k], i,j\in S_a \\
\sum_a d_a(e) & \le c_e & & \forall e\in E\\
d_a(e) & \in [0,1] & & \forall a\in [k], e\in E
\tag{\bf{MCP-LP}} \label{eqn:LP}
\end{align*}
}}
\end{tabular}
\end{small}
\end{center}


In the remainder of this paper we focus exclusively on solutions to
the MCP and CSCP that are collections of vertex partitions. This is
without loss of generality (up to a factor of $2$ in edge loads for
the MCP) and allows us to exploit structural properties of vertex sets
such as laminarity that help in constructing a good approximation.
Accordingly, in the rest of the paper we use the term ``cut'' to
denote a subset of the vertices that defines a vertex partition.



A pair of cuts $C_1,C_2\subset V$ is said to ``cross'' if all of the
sets $C_1\cap C_2$, $C_1\setminus C_2$, and $C_2\setminus C_1$ are
non-empty. A collection $\C=\{C_1,\cdots,C_k\}$ of cuts is said to be
{\em laminar} if no pair of cuts $C_i,C_j\in \C$ crosses. All of our
algorithms are based on the observation that both the MCP and the CSCP
admit near-optimal solutions that are laminar. Specifically, there is
a polynomial-time algorithm that given a fractional feasible solution
to MCP or CSCP (i.e. a feasible solution to \ref{eqn:LP}) produces a
laminar family of fractional cuts that is feasible for the respective
problem and has small load. This is formalized in
Lemmas~\ref{lem:lam1} and \ref{lem:lam2} below. We first introduce the
notion of a fractional laminar family of cuts.

\begin{definition}
A fractional laminar cut family $\C$ for terminal set $T$ with weight
function $w$ is a collection of cuts with the following properties:
\begin{itemize}
\item The collection is laminar
\item Each cut $C$ in the family is associated with a unique terminal
  in $T$. We use $\C_i$ to denote the sub-collection of sets
  associated with terminal $i\in T$. Every $C\in \C_i$ contains the
  node $r_i$.
\item For all $i\in T$, the total weight of cuts in $\C_i$,
  $\sum_{C\in\C_i} w(C)$, is $1$.
\end{itemize}
\end{definition}

Next we define what it means for a fractional laminar family to be
feasible for the MCP or the CSCP. Note that for a terminal pair $i\ne
j$ belonging to the same commodity, condition (2) below is weaker than
requiring cuts in {\em both} $C_i$ and $C_j$ to separate $r_i$ from
$r_j$.

\begin{definition}
\label{def:feas}
A fractional laminar family of cuts $\C$ for terminal set $T$ with
weight function $w$ is feasible for the MCP on a graph $G$ with edge
capacities $c_e$ and commodities $S_1,\cdots,S_k$ if (1) $T=\cup_{a\in
[k]} S_a$, (2) for all $a\in [k]$ and $i,j\in S_a$, $i\ne j$, either
$r_j\not\in \cup_{C\in\C_i} C$, or $r_i\not\in \cup_{C\in\C_j} C$, and
(3) for every edge $e\in E$, $\lef\le c_e$.

The family is feasible for the CSCP on a graph $G$ with edge
capacities $c_e$ and commodities $S_1,\cdots,S_k$ if (1) $T=\cup_{a\in
[k]} S_a \setminus \{t\}$, (2) $t\not\in \cup_{C\in\C} C$, and (3) for
every $e\in E$, $\lef\le c_e$.
\end{definition}

\begin{lemma}
\label{lem:lam1}
Consider an instance of the CSCP with graph $G=(V,E)$, common sink
$t$, edge capacities $c_e$, and commodities $S_1,\cdots,S_k$. Given a
feasible solution $d$ to \ref{eqn:LP}, algorithm {\em Lam-1} produces
a fractional laminar cut family $\C$ that is feasible for the CSCP on
$G$ with edge capacities $c_e+o(1)$.
\end{lemma}

\begin{lemma}
\label{lem:lam2}
Consider an instance of the MCP with graph $G=(V,E)$, edge capacities
$c_e$, and commodities $S_1,\cdots,S_k$. Given a feasible solution $d$
to \ref{eqn:LP}, algorithm {\em Lam-2} produces a fractional laminar
cut family $\C$ that is feasible for the MCP on $G$ with edge
capacities $8c_e+o(1)$.
\end{lemma}

Lemmas~\ref{lem:lam1} and \ref{lem:lam2} are proven in
Section~\ref{sec:laminar}. In Section~\ref{sec:rounding} we show how
to deterministically round a fractional laminar solution to the CSCP
and MCP into an integral one while increasing the load on every edge
by no more than a small additive amount. These rounding algorithms are
the main contributions of our work, and crucially use the laminarity
of the fractional solution.

\begin{lemma}
\label{lem:round1}
Given a fractional laminar cut family $\C$ feasible for the CSCP on a
graph $G$ with {\em integral} edge capacities $c_e$, the algorithm
{\em Round-1} produces an integral family of cuts $\A$ that is
feasible for the CSCP on $G$ with edge capacities $c_e+1$.
\end{lemma}

For the MCP, the rounding algorithm loses an additive factor of $3$ in
edge load.

\begin{lemma}
\label{lem:round2}
Given a fractional laminar cut family $\C$ feasible for the MCP on a
graph $G$ with {\em integral} edge capacities $c_e$, the algorithm
{\em Round-2} produces an integral family of cuts $\A$ that is
feasible for the MCP on $G$ with edge capacities $c_e+3$.
\end{lemma}

Combining these lemmas together we obtain the following theorem.

\begin{theorem}
\label{thm:main}
There exists a polynomial-time algorithm that given an instance of the
MCP with graph $G=(V,E)$, edge capacities $c_e$, and commodities
$S_1,\cdots,S_k$, produces a family $\A$ of multiway cuts, one for
each commodity, such that for each $e\in E$, $\ld{\A}\le
8c_e+4$.

There exists a polynomial-time algorithm that given an instance of the
CSCP with graph $G=(V,E)$, edge capacities $c_e$, and commodities
$S_1,\cdots,S_k$, produces a family $\A$ of multiway cuts, one for
each commodity, such that for each $e\in E$, $\ld{\A}\le c_e+2$.
\end{theorem}

\section{Rounding fractional laminar cut families}
\label{sec:rounding}

In this section we develop algorithms for rounding feasible fractional
laminar solutions to the MCP and the CSCP to integral ones while
increasing edge loads by a small additive amount. We first demonstrate
some key ideas behind the algorithm and the analysis for the CSCP, and
then extend them to the more general case of multiway cuts. Throughout
the section we assume that the edge capacities $c_e$ are integral.

\subsection{The common sink case (proof of Lemma~\ref{lem:round1})}

Our rounding algorithm for the CSCP rounds fractional cuts roughly in
the order of innermost cuts first. The notion of an innermost terminal
is defined with respect to the fractional solution. After each
iteration we ensure that the remaining fractional solution continues
to be feasible for the unassigned terminals and has small edge loads.
We use $\C$ to denote the fractional laminar cut family that we start
out with and $\A$ to denote the integral family that we construct.
Recall that for an edge $e\in E$, $\lef$ denotes the load of the
fractional cut family $\C$ on $e$, and $\lei$ denotes the load of the
integral cut family $\A$ on $e$. We call the former the fractional
load on the edge, and the latter its integral load.

We now formalize what we mean by an ``innermost'' terminal. For every
vertex $v\in V$, let $K_v$ denote the set of cuts in $\C$ that contain
$v$. The ``depth'' of a vertex $v$ is the total weight of all cuts in
$K_v$: $d_v = \sum_{C\in K_v} w(C)$. The depth of a terminal is
defined as the depth of the vertex at which it resides. Terminals are
picked in order of decreasing depth.

Before we describe the algorithm we need some more notation. At any
point during the algorithm we use $S_e$ to denote the set of cuts
crossing an edge $e$. As the algorithm proceeds, the integral loads on
edges increase while their fractional loads decrease. Whenever the
fractional load of an edge becomes $0$, we merge its end-points to
form ``meta-nodes''. At any point of time, we use $M(v)$ to denote the
meta-node containing a node $v\in V$. 

Finally, for a set of fractional cuts $L=\{L_1,\cdots,L_l\}$ with
$L_1\subseteq L_2\subseteq \cdots\subseteq L_l$ and weight function
$w$, we use $L^x$ to denote the subset of $L$ containing the innermost
cuts with weight exactly $x$. That is, let $l'$ be such that
$\sum_{a<l'} w(L_a) < x$ and $\sum_{a\le l'} w(L_a) \ge x$. Then $L^x$
is the set $\{L_1,\cdots,L_{l'}\}$ with weight function $w'$ such that
$w'(L_a)=w(L_a)$ for $a<l'$ and $w'(L_{l'})=x-\sum_{a<l'} w(L_a)$.

\begin{figure*}[t]
  \begin{small} 
  \rule[0.025in]{\textwidth}{0.01in}
  {\bf Input:} Graph $G=(V,E)$ with capacities $c_e$, terminals $T$
  with a fractional laminar cut family $\C$, common sink $t$ with
  $t\not\in \cup_{C\in\C}C$.\\ 
  {\bf Output:} A collection of cuts $\A$, one for each terminal in $T$.\\
  \rule[0.025in]{\textwidth}{0.01in}

  \begin{enumerate}
  \item Initialize $T'=T$, $\A=\emptyset$, and $M(v)=\{v\}$ for all
    $v\in V$. Compute the depths of vertices and terminals.
  \item While there are terminals in $T'$ do:
    \begin{enumerate}
    \item Let $i$ be a terminal with the maximum depth in $T'$. Let
      $A_i=M(r_i)$. Add $A_i$ to $\A$ and remove $i$ from $T'$.
    \item Let $K=K_{r_i}^1$.  Remove cuts in $K\cap \C_i$ from $K$,
      $\C_i$ and $\C$. While there exists a terminal $j\in T'$ with a
      cut $C\in K\cap\C_j$, do the following: let $w=w(C)$; remove $C$
      from $K$, $\C_j$ and $\C$; remove cuts in $\C_i^{w}$ from $\C_i$
      and add them to $\C_j$ (that is, these cuts are reassigned from
      terminal $i$ to terminal $j$).
    \item If there exists an edge $e=(u,v)$ with $\lef=0$, merge the
      meta-nodes $M(u)$ and $M(v)$ (we say that the edge $e$ has been
      ``contracted'').
    \item Recompute the depths of vertices and terminals.
    \end{enumerate}
  \end{enumerate}
  \rule[0.025in]{\textwidth}{0.01in}
  \caption{Algorithm {\em Round-1}---Rounding algorithm for common-sink
  $s$-$t$ cut packing}
  \label{fig:Round1}
  \end{small}
\end{figure*}

The algorithm {\em Round-1} is given in Figure~\ref{fig:Round1}. At
every step, the algorithm picks a terminal, say $i$, with the maximum
depth and assigns an integral cut to it. This potentially frees up
capacity used up by the fractional cuts of $i$, but may use up extra
capacity on some edges that was previously occupied by fractional cuts
belonging to other terminals. In order to avoid increasing edge loads,
we reassign to terminals in the latter set, fractional cuts of $i$
that have been freed up.

Our analysis has two parts. Lemma~\ref{lem:round1-feas} shows that the
family $\C$ continues to remain feasible, that is it always satisfy
the first two conditions in Definition~\ref{def:feas} for the
unassigned terminals. Lemma~\ref{lem:round1-load} analyzes the total
load of the fractional and integral families as the algorithm
progresses.

\begin{lemma}
\label{lem:round1-feas}
Throughout the algorithm, the cut family $\C$ is a fractional laminar
family for terminals in $T'$ with $t\not\in \cup_{C\in\C}C$.
\end{lemma}

\begin{proof}
We prove this by induction over the iterations of the algorithm. The
claim obviously holds at the beginning of the algorithm. Consider a
step at which some terminal $i$ is assigned an integral cut. The
algorithm removes all the cuts in $K=K_{r_i}^1$ from $\C$. Some of
these cuts belong to other terminals; those terminals are reassigned
new cuts. Specifically, we first remove cuts in $K\cap\C_i$ from the
cut family. The total weight of the remaining cuts in $K$ as well as
the total weight of those in $\C_i$ is equal at this time.
Subsequently, we successively consider terminals $j$ with a cut $C\in
K \cap \C_j$, and let $w=w(C)$. Then we remove $C$ from the cut family,
and reassign cuts of total weight $w$ in $\C_i^w$ to $j$. Therefore,
the total weight of cuts assigned to $j$ remains $1$. Furthermore, the
newly reassigned cuts contain the cut $C$, and therefore the vertex
$r_j$, but do not contain the sink $t$. Therefore, $\C$ continues to
be a fractional laminar family for terminals in $T'$.
\end{proof}

\begin{lemma}
\label{lem:round1-load}
At any point of time for every edge $e\in E$, $\lei\le c_e-1$ implies
$\lei+\lef\le c_e$, $\lei=c_e$ implies $\lef\le 1$, and $\lei=c_e+1$
implies $\lef=0$. Furthermore, for $e=(u,v)$, $\lei=c_e$ implies that
either $K_u\cap S_e$ or $K_v\cap S_e$ is empty.
\end{lemma}
\begin{proof}
Let $e=(u,v)$. We prove the lemma by induction over time. Note that in
the beginning of the algorithm, we have for all edges $\lef \leq c_e$
and $\lei =0$, so the inequality $\lei + \lef \leq c_e$ holds.

Let us now consider a single iteration of the algorithm and suppose
that the integral load of the edge increases during this
iteration. (If it doesn't increase, since $\lef$ only decreases over
time, the claim continues to hold.) Let $i$ be the commodity picked by
the algorithm in this iteration, then $M(r_i)$ is the same as either
$M(u)$ or $M(v)$. Without loss of generality assume that $r_i\in
M(u)$. Let $\alpha$ denote the total weight of cuts in $K_u\cap S_e$
and $\beta$ denote the total weight of cuts in $K_v\cap S_e$ prior to
this iteration. Then, $\alpha+\beta=\lef$. Moreover, all cuts in
$\C\setminus S_e$ either contain both or neither of $u$ and $v$. So we
can relate the depths of $v$ and $u$ in the following way: $d_v = d_u
- \alpha + \beta$. Since $i$ is the terminal picked during this
iteration, we must have $d_u\ge d_v$, and therefore, $\alpha\ge
\beta$.

We analyze the final edge load depending on the value of $\alpha$. Two
cases arise: suppose first that $\alpha\ge 1$. Then $K_u^1\subseteq
K_u\cap S_e$, and the fractional weight of $e$ reduces by exactly
$1$. On the other hand, the integral load on the edge increases by
$1$, and so the total load continues to be the same as before. On the
other hand, if $\alpha\le 1$, then $K_u\cap S_e\subseteq K_u^1$, and
all the cuts in $K_u\cap S_e$ get removed from $S_e$ in this
iteration.  Therefore the final fractional load is at most
$\beta\le\alpha\le 1$, and at the end of the iteration, $K_u\cap
S_e=\emptyset$. If $\lei\le c_e-1$, we immediately get that the total
load on the edge is at most $c_e$.

If $\lei=c_e$, then prior to this iteration $\lei=c_e-1$, and so
$\lef\le 1$ by the induction hypothesis. Then, as we argued above,
$\alpha\le\lef\le 1$ implies that the new fractional load on the edge
is at most $1$ and at the end of the iteration, $K_u\cap
S_e=\emptyset$.

Finally, if $\lei=c_e+1$, then prior to this iteration, $\lei=c_e$ and
by the induction hypothesis, $\beta$ is zero (as $\alpha\ge\beta$ and
either $K_u\cap S_e$ or $K_v\cap S_e$ is empty). Along with the fact
that $\alpha\le 1$ (by the inductive hypothesis), the final fractional
load on the edge is $\beta=0$.
\end{proof}

The two lemmas together give us a proof of Lemma~\ref{lem:round1}. We
restate the lemma for completeness.

\setcounter{theorem}{3}
\addtocounter{theorem}{-1}
\begin{lemma}
Given a fractional laminar cut family $\C$ feasible for the CSCP on a
graph $G$ with {\em integral} edge capacities $c_e$, the algorithm
{\em Round-1} produces an integral family of cuts $\A$ that is
feasible for the CSCP on $G$ with edge capacities $c_e+1$.
\end{lemma}
\setcounter{theorem}{7}

\begin{proof}
First note that for every $i$, $A_i$ is set to be the meta-node of
$r_i$ at some point during the algorithm, which is a subset of every
cut in $\C_i$ at that point of time. Then $r_i\in A_i$, and by
Lemma~\ref{lem:round1-feas}, $t\not\in A_i$. Second, for any edge $e$,
its integral load $\lei$ starts out at being $0$ and gradually
increases by at most an additive $1$ at every step, while its
fractional load decreases. Once the fractional load of an edge becomes
zero, both its end points belong to the same meta-node, and so the
edge never gets loaded again. Therefore, by
Lemma~\ref{lem:round1-load}, the maximum integral load on any edge $e$
is at most $c_e+1$.
\end{proof}

\subsection{The general case (proof of Lemma~\ref{lem:round2})}

As in the common-sink case, the rounding algorithm for the MCP
proceeds by picking terminals according to an order suggested by the
fractional solution and assigning the smallest cuts possible to them
subject to the availability of capacity on the edges. In the algorithm
{\em Round-1}, we reassign cuts among terminals at every iteration so
as to maintain the feasibility of the remaining fractional
solution. In the case of MCP, this is not sufficient---a simple
reassignment of cuts as in the case of algorithm {\em Round-1} may not
ensure separation among terminals belonging to the same commodity. We
use two ideas to overcome this difficulty: first, among terminals of
equal depth, we use a different ordering to pick the next terminal to
minimize the need for reassigning cuts; second, instead of reassigning
cuts, we modify the existing fractional cuts for unassigned terminals
so as to remain feasible while paying a small extra cost in edge load.

We now define the ``cut-inclusion'' ordering over terminals.  For
every terminal $i\in T$, let $O_i$ denote the largest (outermost) cut
in $\C_i$, that is, $\forall C\in \C_i$, $C\subseteq O_i$. We say that
terminal $i$ dominates (or precedes) terminal $j$ in the cut-inclusion
ordering, written $i>_{CI} j$, if $O_i\subset O_j$ (if $O_i=O_j$ we
break ties arbitrarily but consistently). Cut-inclusion defines a
partial order on terminals. Note that we can pre-process the cut
family $\C$ by reassigning cuts among terminals, such that for all
pairs of terminals $i, j\in T$ with $i>_{CI} j$, and for all cuts
$C_i\in \C_i$ and $C_j\in\C_j$ with $r_i,r_j\in C_i\cap C_j$, we have
$C_i\subseteq C_j$. We call this property the ``inclusion
invariant''. Ensuring this invariant requires a straightforward
pairwise reassignment of cuts among the terminals, and we omit the
details. Note that following this reassignment, for every terminal
$i$, the new outermost cut of $i$, $O_i$, is the same as or a subset
of its original outermost cut.

As the algorithm proceeds we modify the collection $\C$ as well as
build up the collection $\A$ of integral cuts $A_i$ for $i\in T$. For
example, we may split a cut $C$ into two cuts containing the same
nodes as $C$ and with weights summing to that of $C$. As cuts in $\C$
are modified, their ownership by terminals remains unchanged, and we
therefore continue using the same notation for them. Furthermore, if
for two cuts $C_1$ and $C_2$, we have (for example) $C_1\subseteq C_2$
at the beginning of the algorithm, this relationship continues to hold
throughout the algorithm. This implies that the inclusion invariant
continues to hold throughout the algorithm. We ensure that throughout
the execution of the algorithm the cut family $\C$ continues to be a
fractional laminar family for terminals $T'$. At any point of time,
the depth of a vertex or a terminal, as well as the cut-inclusion
ordering is defined with respect to the current fractional family
$\C$.

As before, let $S_e$ denote the set of cuts in $\C$ that cross $e$ ---
$S_e = \{C\in\C | e\in\delta(C)\}$. Recall that $K_v$ denotes the set
of cuts in $\C$ containing the vertex $v$, and of these $K_v^1$
denotes the inner-most cuts with total weight exactly $1$.

The rounding algorithm is given in Figure~\ref{fig:Round-2}. Roughly
speaking, at every step, the algorithm picks a maximum depth terminal
$i$ and assigns the cut $M(r_i)$ to it (recall that $M(r_i)$ is the
meta-node of the vertex $r_i$ where terminal $i$ resides). It ``pays''
for this cut using fractional cuts in $K_{r_i}^1$. Of course some of
the cuts in $K_{r_i}^1$ belong to other commodities, and need to be
replaced with new fractional cuts. The cut-inclusion invariant ensures
that these other commodities reside at meta-nodes other than $M(r_i)$,
so we modify each cut in $K_{r_i}^1\setminus C_i$ by removing $M(r_i)$
from it (see Figure~\ref{fig:rounding-step}). This process potentially
increases the total loads on edges incident on $M(r_i)$ by small
amounts, but on no other edges. Step~\ref{step:c} of the algorithm
deals with the case in which edges incident on $M(r_i)$ are already
overloaded; In this case we avoid loading those edges further by
assigning to $i$ some subset of the meta-node $M(r_i)$.
Lemmas~\ref{lem:meta-node-1} and \ref{lem:meta-node-2} show that this
case does not arise too often.

\begin{figure}[]
\begin{center}
\epsfig{file = 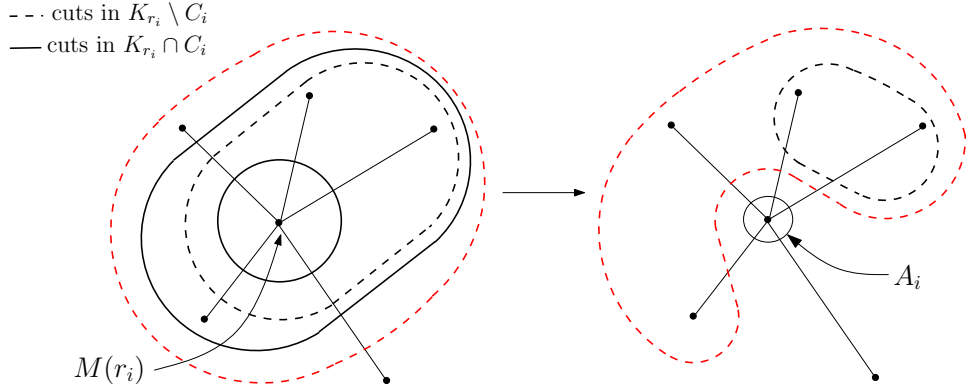, height = 2in}
\caption{An iteration of algorithm {\em Round-2} (Steps 3b \& 3d)}
\label{fig:rounding-step}
\end{center}
\end{figure}

For a terminal $i$ and edge $e$, if at the time that $i$ is picked in
Step~\ref{step:a} of the algorithm $e$ is in $\delta(M(r_i))$, we say
that $i$ accesses $e$. If $e\in E_i$, we say that $i$ defaults on $e$,
and if $e$ is in $\delta(A_i)$ after this iteration, then we say that
$i$ loads $e$.

During the course of the algorithm integral loads on edges increase,
but fractional loads may increase or decrease. To study how these edge
loads change during the course of the algorithm, we divide edges into
five sets. Let $X_{-1}$ denote the set of edges with $\lei\le c_e-1$
and $\lef>0$. For $a\in \{0,1\}$, let $X_a$ denote the set of edges
with $\lei = c_e+a$ and $\lef>0$. $Y$ denotes the set of edges with
$\lei\ge c_e+2$ and $\lef>0$, and $Z$ denotes the set of edges with
$\lef=0$. Every edge starts out with a zero integral load. As the
algorithm proceeds, the edge goes through one or more of the $X_a$s,
may enter the set $Y$, and eventually ends up in the set $Z$. As for
the CSCP, when an edge enters $Z$, we merge the end-points of the edge
into a single meta-node. However, unlike for the CSCP, edges may get
loaded even after entering $Z$. When an edge enters $Y$, we avoid
loading it further (Step~\ref{step:c}), and instead load some edges in
$Z$. Nevertheless, we ensure that edges in $Z$ are loaded no more than
once.


As before our analysis has two components. First we show
(Lemma~\ref{lem:feasible}) that the cuts produced by the algorithm are
feasible. The following lemmas give the desired guarantees on the
edges' final loads: Lemmas~\ref{lem:X-1} and \ref{lem:X} analyze the
loads of edges in $X_a$ for $a\in\{-1,0,1\}$; Lemma~\ref{lem:Y}
analyzes edges in $Y$ and Lemmas~\ref{lem:meta-node-1} and
\ref{lem:meta-node-2} analyze edges in $Z$. We put everything together
in the proof of Lemma~\ref{lem:round2} at the end of this section.

\begin{figure*}[t]
  \begin{small}
  \rule[0.025in]{\textwidth}{0.01in}
  {\bf Input:} Graph $G=(V,E)$ with capacities $c_e$ on edges, a set
  of terminals $T$ with a fractional laminar cut family $\C$.\\
  {\bf Output:} A collection of cuts $\A$, one for each terminal in
  $T$.\\
  \rule[0.025in]{\textwidth}{0.01in}

  \begin{enumerate}
  \item Preprocess the family $\C$ so that it satisfies the inclusion
    invariant.  
  \item Initialize $T'=T$, $\A=\emptyset$, $Y, Z=\emptyset$, and
    $M(v)=\{v\}$ for all $v\in V$.
  \item While there are terminals in $T'$ do:
    \begin{enumerate}
    \item \label{step:a} Consider the set of unassigned terminals with the
      maximum depth, and of these let $i\in T'$ be a terminal that is
      undominated in the cut inclusion ordering. Let $E_i=Y\cap
      \delta(M(r_i))$.
    \item If $E_i=\emptyset$, let $A_i=M(r_i)$.
    \item \label{step:c} If $E_i\ne \emptyset$ (we say that the terminal
      has ``defaulted'' on edges in $E_i$), let $U_i$ denote the set of
      end-points of edges in $E_i$ that lie in $M(r_i)$. If $r_i\in U_i$,
      abort and return error. Otherwise, consider the vertex in $U_i$ that
      entered $M(r_i)$ first during the algorithm's execution, call this
      vertex $u_i$. Set $A_i$ to be the meta-node of $r_i$ just prior to
      the iteration where $M(u_i)$ becomes equal to $M(r_i)$.
    \item Add $A_i$ to $\A$. Remove $\C_i$ from $\C$ and $i$ from
      $T'$. For every $j\in T'$ and $C\in K_{r_i}^1\cap \C_j$, let
      $C=C\setminus\{M(r_i)\}$.
    \item If for some edge $e$, $\lei=c_e+2$ and $\lef>0$, add $e$ to
      $Y$. If there exists an edge $e=(u,v)$ with $\lef=0$, merge the
      meta-nodes $M(u)$ and $M(v)$ (we say that the edge $e$ has been
      ``contracted''.) Add all edges $e$ with $\lef=0$ to $Z$ and remove
      them from $Y$.
    \item Recompute the depths of vertices and terminals.
    \end{enumerate}
  \end{enumerate}
  \rule[0.025in]{\textwidth}{0.01in}
  \caption{Algorithm {\em Round-2}---Rounding algorithm for multiway cut packing}
  \label{fig:Round-2}
  \end{small}
\end{figure*}

\begin{lemma}
\label{lem:feasible}
For all $i$, $r_i\in A_i\subseteq O_i$.
\end{lemma}

\begin{proof}
Each cut $A_i$ is set equal to the meta-node of $r_i$ at some stage of
the algorithm. Therefore, $r_i\in A_i$ for all $i$. Furthermore, at
the time that $i$ is assigned an integral cut, $A_i\subseteq
M(r_i)\subseteq O_i$.
\end{proof}

Next we prove some facts about the fractional and integral loads as an
edge goes through the sets $X_a$. The proofs of the following two
lemmas are similar to that of Lemma~\ref{lem:round1-load}.

\begin{lemma}
\label{lem:X-1}
At any point of time, for every edge $e\in X_{-1}$, $\lei+\lef\le c_e$.
\end{lemma}

\begin{proof}
We prove the claim by induction over time. Note that in the beginning
of the algorithm, we have for all edges $\lef \leq c_e$ and $\lei
=0$, so the inequality $\lei + \lef \leq c_e$ holds.

Let us now consider a single iteration of the algorithm and suppose
that the edge $e$ remains in the set $X_{-1}$ after this step. There
are three events that influence the load of the edge $e=(u,v)$: (1) a
terminal at some vertex in $M(u)$ accesses $e$; (2) a
terminal at $M(v)$ accesses $e$; and, (3) a terminal at some other
meta-node $M\ne M(u), M(v)$ is assigned an integral cut. Let us
consider the third case first, and suppose that a terminal $i$ is
assigned. Since $A_i\subseteq M$ and therefore $e\notin \delta(A_i)$
its integral load does not increase. However, in the event that
$S_e\cap \C_i$ is non-empty, the fractional load on $e$ may decrease
(because cuts in $\C_i$ are removed from $\C$). Therefore, the
inequality continues to hold.

Next we consider the case where a terminal, say $i$, with $r_i\in
M(u)$ accesses $e$ (the second case is similar). Note that
$M(r_i)=M(u)$. In this case the integral load of the edge $e$
potentially increases by $1$ (if the terminal loads the edge). By the
definition of $X_{-1}$, the new integral load on this edge is no more
than $c_e-1$. The fractional load on $e$ changes in three ways:
\begin{itemize}
\item Cuts in $\C_i\cap S_e$ are removed from $\C$, decreasing $\lef$.
\item Some of the cuts in $(K_{r_i}^1\setminus \C_i)\setminus S_e$ get
``shifted'' on to $e$ increasing $\lef$ (we remove the meta-node
$M(r_i)$ from these cuts, and they may continue to contain $M(v)$).
\item Cuts in $(K_{r_i}^1\setminus \C_i)\cap S_e$ get shifted off from
$e$ decreasing $\lef$ (these cuts initially contain $M(r_i)$ but not
$M(v)$, and during this step we remove $M(r_i)$ from these cuts).
\end{itemize}
So the decrease in $\lef$ is at least the total weight of
$K_{r_i}^1\cap S_e=K_u^1\cap S_e$, whereas the increase is at most the
total weight of $K_{r_i}^1\setminus S_e=K_u^1\setminus S_e$.

In order to account for the two terms, let $\alpha$ denote the total
weight of cuts in $K_u\cap S_e$, and $\beta$ denote the total weight
of cuts in $K_v\cap S_e$. Then, $\alpha+\beta=\lef$. As in the proof
of Lemma~\ref{lem:round1-load}, we have $d_v = d_u-\alpha+\beta$, and
therefore $d_u\ge d_v$ implies $\alpha\ge \beta$. Now, suppose that
$\alpha\ge 1$. Then $K_u^1\subseteq S_e$. Therefore, the decrease in
$\lef$ due to the sets $K_u^1\cap S_e = K_u^1$ is at least $1$, and
there is no corresponding increase, so the sum $\lei+\lef$ remains at
most $c_e$.

Finally, suppose that $\alpha<1$. Then $K_u^1$ contains all the cuts
in $K_u\cap S_e$, the weight of $K_u^1\cap S_e$ is exactly $\alpha$,
and so the decrease in $\lef$ is at least $\alpha$. Moreover, the
total weight of $K_u^1\setminus S_e$ is $1-\alpha$, therefore, the
increase in $\lef$ due to the sets in $K_u^1\setminus S_e$ is at most
$1-\alpha$. Since $\lef$ starts out as being equal to $\alpha+\beta$,
its final value after this step is $1-\alpha+\beta \le 1$ as $\beta\le
\alpha$. Noting that $\lei$ is at most $c_e-1$ after the step, we get
the desired inequality.
\end{proof}

\begin{lemma}
\label{lem:X}
For any edge $e=(u,v)$, from the time that $e$ enters $X_0$ to the
time that it exits $X_1$, $\lef\le 1$. Furthermore suppose (without
loss of generality) that during this time in some iteration $e$ is
accessed by a terminal $i$ with $r_i\in M(u)$, then following this
iteration until the next time that $e$ is accessed, we have $S_e\cap
K_u=\emptyset$, and the next access to $e$ (if any) is from a terminal
in $M(v)$.
\end{lemma}

\begin{proof}
First we note that if the lemma holds the first time an edge $e=(u,v)$
enters a set $X_a$, $a\in\{0,1\}$, then it continues to hold while the
edge remains in $X_a$. This is because during this time the integral
load on the edge does not increase, and therefore throughout this time
we assign integral cuts to terminals at meta-nodes different from
$M(u)$ and $M(v)$ --- this only reduces the fractional load on the
edge $e$ and shrinks the set $S_e$.

Consider the first time that an edge $e=(u,v)$ moves from the set
$X_{-1}$ to $X_0$. Suppose that at this step we assign an integral cut
to a terminal $i$ residing at node $r_i\in M(u)$. Prior to this step,
$\lei=c_e-1$, and so by Lemma~\ref{lem:X-1}, $\lef\le 1$. As before
define $\alpha$ to be the total weight of cuts $K_u\cap S_e$, and
$\beta$ to be the total weight of cuts $K_v\cap S_e$. Then following
the same argument as in the proof of Lemma~\ref{lem:X-1}, we conclude
that the final fractional weight on $e$ is at most $\beta+1-\alpha \le
1$. Furthermore, since $K_u\cap S_e\subseteq K_u^1$, we either remove
all these cuts from $\C$ or shift them off of edge $e$. Moreover, any
new cuts that we shift on to $e$ do not contain the meta-node
$M(r_i)=M(u)$, and in particular do not contain the vertex
$u$. Therefore at the end of this step, $S_e\cap K_u=\emptyset$. This
also implies that following this iteration terminals in $M(v)$ have
depth larger than terminals in $M(u)$, and so the next access to $e$
must be from a terminal in $M(v)$.

The same argument works when an edge moves from $X_0$ to $X_1$. We
again make use of the fact that prior to the step the fractional load
on the edge is at most $1$.
\end{proof}

\begin{lemma}
\label{lem:Y}
During any iteration of the algorithm, for any edge $e\in Y$, the
following are satisfied:
\begin{itemize}
\item $\lef\le 1$
\item If the edge $e=(u,v)$ is accessed by a terminal $i$ with $r_i\in
  M(u)$, then following this iteration until the next time that $e$ is
  accessed, we have $S_e\cap K_u=\emptyset$, and the next access to
  $e$ (if any) is from a terminal in $M(v)$.
\item If a terminal $i$ with $r_i\in M(u)$ accesses $e=(u,v)$, then
  $r_i\ne u$, $A_i\cap\{u,v\}=\emptyset$, and so $i$ does not load
  $e$. Also, consider any previous access to the edge by a terminal in
  $M(u)$; then prior to this access, $r_i\not\in M(u)$.
\end{itemize}
\end{lemma}
\begin{proof}
The first two parts of this lemma extend Lemma~\ref{lem:X} to the case
of $e\in Y$, and are otherwise identical to that lemma. The proof for
these claims is analogous to the proof of Lemma~\ref{lem:X}. The only
difference is that terminals accessing an edge $e\in Y$ default on
this edge. However, this does not affect the argument: when a terminal
defaults on the edge, the edge's fractional load changes in the same
way as if the terminal did not default; the only change is in the way
an integral cut is assigned to the terminal. Since these claims depend
only on how the fractional load on the edge changes, they continue to
hold while the edge is in $Y$.

For the third part of the lemma, since $A_i\subseteq M(r_i)=M(u)$ and
$v\not\in M(u)$, $v\not\in A_i$. Next we show that $u\not\in
A_i$. Consider the iterations of the algorithm during which $\lef\le
1$. During this time the edge was accessed at least twice prior to
being accessed by $i$ (once when $e$ moved from $X_0$ to $X_1$, once
when $e$ moved from $X_1$ to $Y$, and possibly multiple times while
$e\in Y$). Let the last two accesses be by the terminals $j_1$ and
$j_2$, at iterations $t_1$ and $t_2$, $t_1\le t_2$. For $a\in
\{0,1\}$, let $M^a(u)$ and $M^a(v)$ denote the meta-nodes of $u$ and
$v$ respectively just prior to iteration $t_a$, and $M(u)$ and $M(v)$
denote the respective meta-nodes just prior to the current
iteration. Then by Lemma~\ref{lem:X} and the second part of this
lemma, we have $r_{j_1}\in M^1(u)$ and $r_{j_2}\in M^2(v)$. We claim
that $i>_{CI} j_2 >_{CI} j_1$. Given this claim, if $r_i\in
M^1(u)=M^1(r_{j_1})$, then since $i$ and $j_1$ have the same depth at
iteration $t_1$, we get a contradiction to the fact that the algorithm
picks $j_1$ before $i$ in Step~\ref{step:a}. Therefore, $r_i\not\in
M(u)$ at any iteration prior to $t_1$, and in particular, $r_i\ne
u$. Finally, since $u\in U_i$ and $U_i\cap A_i=\emptyset$, this also
implies that $u\not\in A_i$.

It remains to prove the claim. We will prove that $j_2>_{CI} j_1$. The
proof for $i>_{CI} j_2$ is analogous. In fact we will prove a stronger
statement: between iterations $t_1$ and $t_2$, all terminals with cuts
in $S_e$ dominate $j_1$ in the cut-inclusion ordering. We prove this
by induction. By Lemma~\ref{lem:X}, prior to iteration $t_1$, $S_e$
does not contain any cuts belonging to terminals at $M(v)$. Following
the iteration, $S_e$ only contains fractional cuts in $K_u^1$ that got
shifted on to the edge $e$. Prior to shifting, these cuts contain
$M^1(u)$, and therefore $r_{j_1}$, but do not belong to $j_1$. Then,
these cuts are subsets of $O_{j_1}$, and so by the inclusion
invariant, they belong to terminals dominating $j_1$ in the
cut-inclusion ordering. Therefore, the claim holds right after the
iteration $t_1$. Finally, following the iteration until the next time
that $e$ is accessed (by $j_2$), the set $S_e$ only shrinks, and so
the claim continues to hold.
\end{proof}

In order to analyze the loading of edges in $Z$, we need some more
notation. Let $\M$ denote the collection of sets of vertices that were
meta-nodes at some point during the algorithm. For any edge $e\in Z$,
let $M_e$ denote the meta-node formed when $e$ enters $Z$; then $M_e$
is the smallest set in $\M$ containing both the end points of
$e$. Note that the collection $\A\cup \M$ is laminar.

\begin{lemma}
\label{lem:meta-node-1}
An edge $e\in Z$ is loaded only if after the formation of $M_e$ a
terminal residing at a vertex in $M_e$ defaults on an edge in
$\delta(M_e)$. (Note that this may happen after $M_e$ has merged with
some other meta-nodes.)
\end{lemma}

\begin{proof}
Let $i$ be a defaulting terminal that loads the edge $e\in Z$. Then
$e\in \delta(A_i)$, and therefore, $A_i\subsetneq M_e$ and $r_i\in
M_e$. Furthermore, since $A_i$ is a strict subset of $M_e$, $U_i\cap
M_e\ne\emptyset$, and therefore, $i$ defaults on an edge $e'\in Y$
with at least one end-point in $M_e$. But if both the end-points of
$e'$ are in $M_e$, then we must have $\ell_{e'}^{\C}=0$ contradicting
the fact that $e'$ is in $Y$. Therefore, $e'\in\delta(M_e)$.
\end{proof}

\begin{lemma}
\label{lem:meta-node-2}
For any meta-node $M\in\M$, after its formation, at most one terminal
residing at a vertex in $M$ can default on edges in $\delta(M)$ (even
after $M$ has merged with other meta-nodes).
\end{lemma}

\begin{proof}
For the sake of contradiction, suppose that two terminals $i$ and $j$,
both residing at vertices in $M$ default on edges in $\delta(M)$ after
the formation of $M$, with $i$ defaulting before $j$. Let $M_1$
($M_2$) denote the meta-node containing $M$ just before $i$ ($j$)
defaulted. Note that $M\subseteq M_1\subseteq M_2$. Consider an edge
$e\in E_j\cap\delta(M)$ (recall that $E_j$ is the set of edges that
$j$ defaults on, so this set is non-empty by our assumption). Then
$e\in \delta(M)\cap \delta(M_2)\subseteq \delta(M_1)$. Therefore, at
the time that $i$ defaulted, $e$ was accessed by $i$, and by the third
claim in Lemma~\ref{lem:Y}, $r_j\not\in M_1$. This contradicts the
fact that $r_j\in M$.
\end{proof}

Finally we can put all these lemmas together to prove our main result
on algorithm {\em Round-2}.

\setcounter{theorem}{4}
\addtocounter{theorem}{-1}
\begin{lemma}
Given a fractional laminar cut family $\C$ feasible for the MCP on a
graph $G$ with {\em integral} edge capacities $c_e$, the algorithm
{\em Round-2} produces an integral family of cuts $\A$ that is
feasible for the MCP on $G$ with edge capacities $c_e+3$.
\end{lemma}
\setcounter{theorem}{13}

\begin{proof}
We first note that the third part of Lemma~\ref{lem:Y} implies that
for all $i$, $r_i\not\in U_i$, and therefore the algorithm never
aborts. Then Lemma~\ref{lem:feasible} implies that we get a feasible
cut packing. Finally, note that every edge starts out in the set
$X_{-1}$, goes through one or more of the $X_a$'s, $a\in \{0,1\}$,
potentially goes through $Y$, and ends up in $Z$. An edge $e$ enters
$Y$ when its integral load becomes $c_e+2$. Lemma~\ref{lem:Y} implies
that edges in $Y$ never get loaded, and so at the time that an edge
$e$ enters $Z$, $\lei\le c_e+2$. After this point the edge stays in
$Z$, and Lemmas~\ref{lem:meta-node-1} and \ref{lem:meta-node-2} imply
that it gets loaded at most once. Therefore, the final load on the
edge is at most $c_e+3$.
\end{proof}

\section{Constructing fractional laminar cut packings}
\label{sec:laminar}

We now show that fractional solutions to the program~\ref{eqn:LP} can
be converted in polynomial time into fractional laminar cut families
while losing only a small factor in edge load. We begin with the
common sink case.

\begin{figure*}[t]
  \begin{small} 
  \rule[0.025in]{\textwidth}{0.01in}
  {\bf Input:} Graph $G=(V,E)$ with edge capacities $c_e$, commodities
  $S_1,\cdots,S_k$, common sink $t$, a feasible solution $d$ to the
  program~\ref{eqn:LP}.\\
  {\bf Output:} A fractional laminar family of cuts $\C$ that is
  feasible for $G$ with edge capacities $c_e+o(1)$.\\
  \rule[0.025in]{\textwidth}{0.01in}

  \begin{enumerate}
    \item \label{step:lam1-a} For every $a\in [k]$ and terminal $i\in
      S_a$ do the following: Order the vertices in $G$ in increasing
      order of their distance under $d_a$ from $r_i$. Let this
      ordering be $v_0=r_i, v_1, \cdots, v_n$. Let $\C_i$ be the
      collection of cuts $\{v_0, v_1, \cdots, v_b\}$, one for each
      $b\in [n]$, $d_a(r_i,v_b)<1$, with weights $w(\{v_0, \cdots,
      v_b\}) = d_a(r_i,v_{b+1}) - d_a(r_i,v_b)$. Let $\C$ denote the
      collection $\{\C_i\}_{i\in \cup_a S_a}$.
    \item \label{step:lam1-a-2} Let $N=nk$. Round up the weights of
      all the cuts in $\C$ to multiples of $1/N^2$, and truncate the
      collection so that the total weight of every sub-collection
      $\C_i$ is exactly $1$. Also split every cut with weight more
      than $1/N^2$ into multiple cuts of weight exactly $1/N^2$ each,
      assigned to the same commodity.
    \item \label{step:lam1-b} While there are pairs of cuts in $\C$
      that cross, consider any pair of cuts $C_i,C_j\in \C$ belonging
      to terminals $i\ne j$ that cross each other. Transform these
      cuts into new cuts for $i$ and $j$ according to
      Figure~\ref{fig:cscp-lam}.
  \end{enumerate}

  \rule[0.025in]{\textwidth}{0.01in}
  \begin{center}
  \caption{Algorithm {\em Lam-1}---Algorithm to convert an LP solution
    for the CSCP into a feasible fractional laminar family}
  \end{center}
  \label{fig:Lam1}
  \end{small}
\end{figure*}

\begin{figure}[]
\begin{center}
\epsfig{file = 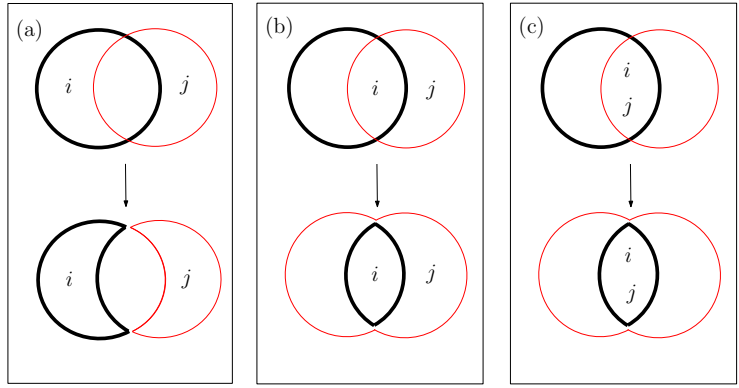, height = 2.0in}
\caption{Rules for transforming an arbitrary cut family into a laminar
  one for the CSCP. The dark cuts in this figure correspond to the
  terminal $i$, and the light cuts to terminal $j$; $t$ lies outside
  all the cuts.}
\label{fig:cscp-lam}
\end{center}
\end{figure}

\subsection{Obtaining laminarity in the common sink case}
\label{sec:lam1}

We prove Lemma~\ref{lem:lam1} in this section. Our algorithm involves
starting with a solution to \ref{eqn:LP}, converting it into a
feasible fractional {\em non-laminar} family of cuts, and then
resolving pairs of crossing cuts one at a time by applying the rules
in Figure~\ref{fig:cscp-lam}. The algorithm is given in
Figure~\ref{fig:Lam1}.

\setcounter{theorem}{1}
\addtocounter{theorem}{-1}
\begin{lemma}
Consider an instance of the CSCP with graph $G=(V,E)$, common sink
$t$, edge capacities $c_e$, and commodities $S_1,\cdots,S_k$. Given a
feasible solution $d$ to \ref{eqn:LP}, algorithm {\em Lam-1} produces
in polynomial time a fractional laminar cut family $\C$ that is
feasible for the CSCP on $G$ with edge capacities $c_e+o(1)$.
\end{lemma}
\setcounter{theorem}{13}

\begin{proof}
We first note that the family $\C$ is feasible for the given instance
of CSCP at the end of Step~\ref{step:lam1-a-2}, but is not necessarily
laminar. Since the number of distinct cuts in $\C$ after
Step~\ref{step:lam1-a} is at most $nk=N$, at the end of
Step~\ref{step:lam1-a-2}, edge loads are at most $c_e+1/N$. As we
tranform the cuts in Step~\ref{step:lam1-b}, we maintain the property
that no cut $C\in\C$ contains the sink $t$, but every cut $C\in\C_i$
contains the node $r_i$ for terminal $i$. It is also easy to see from
Figure~\ref{fig:cscp-lam} that the load on every edge stays the
same. Finally, in every iteration of this step, the number of pairs of
crossing cuts strictly decreases. Therefore, the algorithm ends after
a polynomial number of iterations.
\end{proof}

\subsection{Obtaining laminarity in the general case}
\label{sec:int-lam2}

Obtaining laminarity in the general case involves a more careful
selection and ordering of rules of the form given in
Figure~\ref{fig:cscp-lam}. The key complication in this case is that
we must maintain separation of every terminal from every other
terminal in its commodity set. We first show how to convert an
integral collection of cuts feasible for the MCP into a feasible
integral laminar collection of cuts. We lose a factor of $2$ in edge
loads in this process (see Lemma~\ref{lem:int-lam2} below). Obtaining
laminarity for an arbitrary fractional solution requires converting it
first into an integral solution for a related cut-packing problem and
then applying Lemma~\ref{lem:int-lam2} (see algorithm {\em Lam-2} in
Figure~\ref{fig:Lam2} and the proof of Lemma~\ref{lem:lam2} following
it).

\begin{lemma}
\label{lem:int-lam2}
Consider an instance of the MCP with graph $G=(V,E)$and commodities
$S_1,\cdots,S_k$, and let $\C^1=\{C^1_i\}_{i\in S_a,a\in [k]}$ be a
family of cuts such that for each $a\in [k]$ and $i\in S_a$, $C^1_i$
contains $i$ but no other $j\in S_a$. Then algorithm {\em
Integer-Lam-2} produces a {\em laminar} cut collection
$\C^2=\{C^2_i\}_{i\in S_a,a\in [k]}$ such that for each $a\in [k]$ and
$i\ne j\in S_a$, either $C^2_i$ or $C^2_j$ separates $i$ from $j$, and
$\ld{\C^1}\le 2\ld{\C^2}$ for every edge $e\in E$.
\end{lemma}

In the remainder of this section we interpret cuts as sets of vertices
as well as sets of terminals residing at those vertices. The algorithm
for laminarity in the integral case is given in
Figure~\ref{fig:Int-Lam2}.

\begin{figure*}[]
  \begin{small} 
  \rule[0.025in]{\textwidth}{0.01in}
  {\bf Input:} Graph $G=(V,E)$ with edge capacities $c_e$, commodities
  $S_1,\cdots,S_k$, a family of cuts $\C$ with one cut for every
  terminal in $\cup_a S_a$, such that the cut for terminal $i\in S_a$
  does not contain any terminal $j\ne i$ in $S_a$.\\
  {\bf Output:} A laminar collection of cuts, one for each terminal in
  $\cup_a S_a$, such that for all $a$ and for all $i, j \in S_a$,
  $i\ne j$, either the cut for $i$ or the cut for $j$ separates $i$
  from $j$.\\
  \rule[0.025in]{\textwidth}{0.01in}

  \begin{enumerate}
    \item \label{step:lam-iter} While there are pairs of cuts in $\C$
      that cross, do (see Figure~\ref{fig:int-lam-cases}):
      \begin{enumerate}
	\item \label{step:lam-a} Consider any pair of cuts $C_i,C_j\in
	  \C$ belonging to terminals $i\ne j$ that cross each other,
	  such that $r_i\in C_i\setminus C_j$ and $r_j\in C_j\setminus
	  C_i$. Reassign $C_i = C_i\setminus C_j$ and
	  $C_j=C_j\setminus C_i$. Return to Step~\ref{step:lam-iter}.
	\item \label{step:lam-a'} Consider any three terminals $i_1,
	  i_2, i_3$ with cuts $C_1, C_2$ and $C_3$ such that
	  $r_{i_1}\in C_1\cap C_2\setminus C_3$, $r_{i_2}\in C_2\cap
	  C_3\setminus C_1$, and $r_{i_3}\in C_3\cap C_1\setminus
	  C_2$. Then, reassign these respective intersections to the
	  three terminals. Return to Step~\ref{step:lam-iter}.
	\item \label{step:lam-b} Consider any pair of cuts $C_i,C_j\in
	  \C$ belonging to terminals $i,j\in S_a$ for some $a$ that
	  cross each other, such that $r_i\in C_i\cap C_j$ and $r_j\in
	  C_j\setminus C_i$.  Reassign $C_i = C_i\cap C_j$ and
	  $C_j=C_i\cup C_j$. Return to Step~\ref{step:lam-iter}.
	\item \label{step:lam-c} Consider any pair of cuts $C_i,C_j\in
	  \C$ belonging to terminals $i\ne j$ that cross each other,
	  such that $r_i,r_j\in C_i\cap C_j$, $i\in S_a$ and $j\in
	  S_b$ with $a\ne b$.
	  \begin{itemize}
	  \item Suppose that there is no $i'\in S_a\cap C_j$ with
	    $C_i\subset C_{i'}$. Then, reassign $C_i=C_i\cup C_j$ and
	    $C_j=C_i\cap C_j$; return to Step~\ref{step:lam-iter}.
	    Conversely, if there is no $j'\in S_b\cap C_i$ with
	    $C_j\subset C_{j'}$. Then, reassign $C_j=C_i\cup C_j$ and
	    $C_i=C_i\cap C_j$; return to
	    Step~\ref{step:lam-iter}. (This transformation is similar
	    to Step~\ref{step:lam-b}.)
	  \item If neither of those cases hold, let $i_0=i$, and let
	    $i_1,\cdots,i_x$ denote the terminals in $S_a\cap C_j$
	    with $C_i\subset C_{i_1}\subset C_{i_2} \subset \cdots
	    \subset C_{i_x}$. For $x'\le x-2$, reassign $C_{i_{x'}} =
	    (C_{i_{x'+1}}\setminus C_j) \cup C_{i_{x'}}$,
	    $C_{i_{x-1}}= C_{i_x}\cup C_j$, and $C_{i_x} = C_{i_x}
	    \cap C_j\setminus C_{i_{x-1}}$. Reassign cuts to $j$ and
	    terminals in $S_b\cap C_i$ likewise. Return to
	    Step~\ref{step:lam-iter}.
	  \end{itemize}
	\item If none of the above rules match, then go to
	  Step~\ref{step:lam-d-1}.
      \end{enumerate}
    \item \label{step:lam-d-1} Let $\G$ be a directed graph on the
      vertex set $\cup_a S_a$, with edges colored red or blue, defined
      as follows: for terminals $i\ne j$, $\G$ contains a red edge
      from $i$ to $j$ if and only if $C_j\subset C_i$, and contains a
      blue edge from $i$ to $j$ if and only if $r_j\in C_i$,
      $r_i\not\in C_j$, and $C_j\setminus C_i\ne \emptyset$. We note
      that since no pair of terminals $i$ and $j$ matches the rules in
      Step~\ref{step:lam-iter}, whenever $C_i$ and $C_j$ intersect
      $\G$ contains an edge between $i$ and $j$.

      While there is a directed blue cycle in $\G$, consider the
      shortest such cycle $i_1\rightarrow i_2\rightarrow \cdots
      \rightarrow i_x\rightarrow i_1$. For $x'\le x$, $x'\ne 1$,
      assign to $i_{x'}$ the cut $C_{i_{x'}}\cap C_{i_{x'-1}}$, and
      assign to $i_1$ the cut $C_{i_1}\cap C_{i_x}$.
    \item \label{step:lam-d-2} We show in
      Lemma~\ref{lem:lam2-laminarity} that at this step $\G$ is
      acyclic. For every connected component in $\G$ do:
      \begin{enumerate}
      \item Let $T$ be the set of terminals in the component and $A$
	be the set of corresponding cuts. Assign capacities $p_e =
	2\ld{A}$ to edges in $G$. Let $G_p$ be the graph obtained by
	merging all pairs of vertices that have an edge $e$ with
	$p_e=0$ between them. We call the vertices of $G_p$
	``meta-nodes'' (note that these are sets of vertices in the
	original graph). At any point of time, let $R_i$ denote the
	meta-node at which a terminal $i$ resides.
      \item \label{step:lam-d-2b} While there are terminals in $T$,
	pick any ``leaf'' terminal $i$ (that is, a terminal with no
	outgoing red or blue edges in $\G$). Reassign to $i$ the cut
	$R_i$. Reduce the capacity of every edge $e\in\delta(R_i)$ by
	$1$. Remove $i$ from $T$; remove $i$ and all edges incident on
	it from $\G$. Recompute the graph $G_p$ based on the new
	capacities.
      \end{enumerate}
  \end{enumerate}
  \rule[0.025in]{\textwidth}{0.01in}
  \caption{Algorithm {\em Integer-Lam-2}---Algorithm to convert an
  integral family of multiway cuts into a laminar one}
  \label{fig:Int-Lam2}
  \end{small}
\end{figure*}

\begin{figure*}[t]
\begin{center}
\epsfig{file = 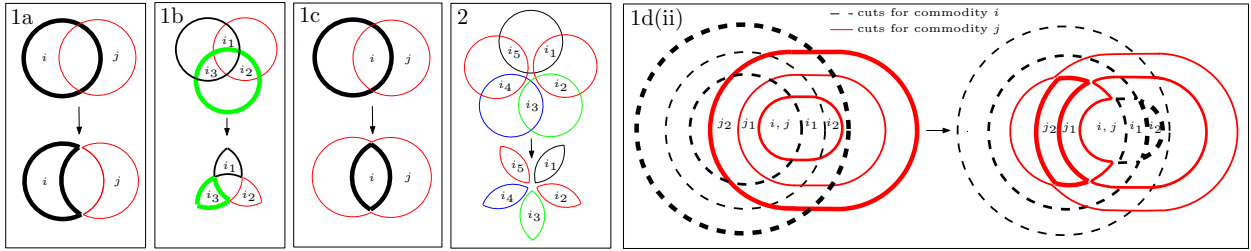, height = 1.3in}
\caption{Some simple rules for resolving crossing cuts. See algorithm
{\em Integer-Lam-2} in Figure~\ref{fig:Int-Lam2} for formal
descriptions.}
\label{fig:int-lam-cases}
\end{center}
\end{figure*}

As in the common sink case, the algorithm starts by applying a series
of simple rules to pairs of crossing cuts while maintaining the
invariant that pairs of terminals belonging to the same commodity are
always separated by at least one of the two cuts assigned to
them. Certain kinds of crossings of cuts are easy to resolve while
maintaining this invariant (Step~\ref{step:lam-iter} of the algorithm
resolves these crossings; see also Figure~\ref{fig:int-lam-cases}). In
Steps~\ref{step:lam-d-1} and \ref{step:lam-d-2}, we ignore the
commodities that each terminal belongs to, and assign new laminar cuts
to terminals while ensuring that the new cut of each terminal lies
within its previous cut (and therefore, separation continues to be
maintained). These steps incur a penalty of $2$ in edge loads.

The rough idea behind Steps~\ref{step:lam-d-1} and \ref{step:lam-d-2}
is to consider the set of all ``conflicting'' terminals, call it
$F$. Then we can assign to each terminal $i\in F$ the cut $\cap_{j\in
F} \hat{C}_j$ where $\hat{C}_j$ is either the cut of terminal $j$ or
its complement depending on which of the two contains $r_i$. These
intersections are clearly laminar, and are subsets of the original
cuts assigned to terminals. Furthermore, if each terminal gets a
unique intersection, then edge loads increase by a factor of at most
$2$. Unfortunately, some groups of terminals may share the same
intersections. In order to get around this, we assign cuts to
terminals in a particular order suggested by the structure of the
conflict graph on terminals (graph $\G$ in the algorithm) and assign
appropriate intersections to them while explicitly ensuring that edge
loads increase by a factor of no more than $2$.

Throughout the algorithm, every terminal in $\cup_a S_a$ has an
integral cut assigned to it. The proof of Lemma~\ref{lem:int-lam2} is
established in three parts:  Lemma~\ref{lem:lam2-laminarity}
establishes the laminarity of the output cut family,
Lemma~\ref{lem:lam2-separation} argues separation, and
Lemma~\ref{lem:lam2-load} analyzes edge loads.

\begin{lemma}
\label{lem:lam2-laminarity}
Algorithm {\em Integer-Lam-2} runs in polynomial time and produces a
laminar cut collection.
\end{lemma}
\begin{proof}
As in the previous section define the crossing number of a family of
cuts to be the number of pairs of cuts that cross each other. We first
note that in every iteration of Steps~\ref{step:lam-iter} and
\ref{step:lam-d-1} of the algorithm, the crossing number of the cut
family $\C$ strictly decreases: no new crossings are created in these
steps, while the crossings of the two or more cuts involved in each
transformation are resolved (see
Figure~\ref{fig:int-lam-cases}). Therefore, after a polynomial number
of steps, we exit Steps~\ref{step:lam-iter} and \ref{step:lam-d-1} and
go to Step~\ref{step:lam-d-2}.

Next, we claim that during Step~\ref{step:lam-d-2} of the algorithm
the graph $\G$ is acyclic. This implies that while $\G$ is non-empty,
we can always find a leaf terminal in Step~\ref{step:lam-d-2};
therefore every terminal in $\G$ gets assigned a new cut. It is
immediate that the graph does not contain any directed blue cycles or
any directed red cycles (the latter follows because red edges define a
partial order over terminals). Suppose the graph contains three
terminals $i_1$, $i_2$ and $i_3$ with a red edge from $i_1$ to $i_2$,
and a red or blue edge from $i_2$ to $i_3$, then it is easy to see
that there must be a red or blue edge from $i_1$ to $i_3$. Therefore,
any multi-colored directed cycle must reduce to either a smaller blue
cycle or a cycle of length $2$. Neither of these cases is possible
(the latter is ruled out by definition), and therefore the graph
cannot contain any multi-colored cycles.

Now consider cuts assigned during Step~\ref{step:lam-d-2}. Let $T$ be
the set of terminals corresponding to some component in $\G$ and
$j\not\in T$. Then before $T$ is processed, $j$'s cut is laminar with
respect to all the cuts in $A_T$, and is therefore a subset of some
meta-node in $G_{p^T}$. So the new cuts assigned to terminals in $T$
are also laminar with respect to $j$'s cut.

Finally, consider any two cuts assigned during Step~\ref{step:lam-d-2}
of the algorithm and belonging to two terminals in the same component
of $\G$. Consider the set of all meta-nodes created during this
iteration of Step~\ref{step:lam-d-2}. This set is laminar, and the
cuts assigned during this iteration are a subset of this laminar
family. Therefore, they are laminar.
\end{proof}

\begin{lemma}
\label{lem:lam2-inclusion}
For a commodity $i$ assigned a cut in Step~\ref{step:lam-d-2} of
algorithm {\em Integer-Lam-2}, let $C^1_i$ be its cut before this
step, and $C^2_i$ be the new cut assigned to it. Then $C^2_i\subseteq
C^1_i$.
\end{lemma}
\begin{proof}
We assume without loss of generality that prior to
Step~\ref{step:lam-d-2} each edge load is at most one; this can be
achieved by splitting a multiply-loaded edge into many edges. We focus
on the behavior of the algorithm for a single component $T$ of $\G$
and prove the lemma by induction over time.

Consider an iteration of Step~\ref{step:lam-d-2b} during which some
terminal $i\in T$ is assigned and let $C_i$ be its original
cut. Consider any vertex $v\not\in C_i$ and let $P$ be a shortest
simple path from $r_i$ to $v$ in $G_{p^T}$ (where the length of an
edge $e$ is given by $p^T_e$ just prior to when $i$ is assigned a new
cut). It is easy to see that there is one such shortest path that
crosses each new cut assigned prior to this iteration in
Step~\ref{step:lam-d-2b} at most twice -- suppose there are multiple
entries and exits for some cut, we can ``short-cut'' the path by
connecting the first point on the path inside the cut to the last
point on the path inside the cut via a simple path of length $0$ lying
entirely inside the cut. We pick $P$ to be such a path. We will prove
that $P$'s length is at least $2$. So the meta-node containing $i$
must lie inside the cut $C_i$, and the lemma holds.

Let $T_1$ (resp. $T_2$) be the set of terminals in $T\setminus C_i$
(resp. $T\cap C_i$) that are assigned new cuts before $i$ in this
iteration. We first note that for any $j$ in $T_1$, prior to this
step, there is no edge from $j$ to $i$ (as $j$ is assigned before
$i$), so $r_i\not\in C_j$, and this along with $r_j\not\in C_i$
implies that $C_i$ and $C_j$ are disjoint. This implies that the new
cut of $j$ (which is a subset of $C_j$ by induction) is also disjoint
from $C_i$, and therefore cannot load any edge with an end-point in
$C_i$. So the only new cuts assigned this far in
Step~\ref{step:lam-d-2b} that load edges in $P$ belong to terminals in
$T_2$.

Now we will analyze $P$'s length by accounting for all the newly
assigned cuts that load its edges. Let $S_P$ be the set of terminals
in $T_2$ that load an edge in $P$, and $j\in S_P$. Since the new cut
of $j$ intersects $P$, by the induction hypothesis, $C_j$ should
either intersect $P$ or contain the entire path inside it. If $C_j$
contains $P$ entirely, then $C_j\setminus C_i\ne\emptyset$, and
furthermore $r_i,r_j\in C_i\cap C_j$. This implies that either
$C_i\subset C_j$ and there is a directed red edge from $j$ to $i$, or
$C_i\setminus C_j\ne \emptyset$, that is, $C_i$ and $C_j$ cross and
should have matched the rule in Step~\ref{step:lam-c} of the
algorithm. Both possibilities lead to a contradiction. Therefore,
$C_j$ must intersect $P$.

Finally, the original total length of the path is at least $2|S_P|+2$,
because each terminal in $S_P$ contributes two units towards its
length, and another two units is contributed by $C_i$. Out of these up
to $2|S_P|$ units of length is consumed by terminals in
$S_P$. Therefore, at the time that $i$ is assigned a cut, at least $2$
units remain.
\end{proof}

\begin{lemma}
\label{lem:lam2-separation}
When algorithm {\em Integer-Lam-2} terminates, for every $a\in [k]$
and $i\ne j\in S_a$, either $C_i$ or $C_j$ separates $i$ from $j$.
\end{lemma}
\begin{proof}
We claim that for every $a\in [k]$ and $i\ne j\in S_a$, at every time
step during the execution of the algorithm, $|C_i\cap C_j
\cap\{r_i,r_j\}|\le 1$. Then since by Lemma~\ref{lem:lam2-laminarity}
the final solution is laminar, the lemma follows. We prove this claim
by induction over time. First, if during any iteration of the
algorithm, we ``shrink'' the cut of any terminal (that is, reassign to
the terminal a cut that is a strict subset of its original cut), then
the claim continues to hold for that terminal, because intersections
of the terminal's cut only shrink in that step. Note that cuts of
terminals expand only in Steps~\ref{step:lam-b} and \ref{step:lam-c}
of the algorithm (by construction and by
Lemma~\ref{lem:lam2-inclusion}).

Suppose that during some iteration we apply the transformation in
Step~\ref{step:lam-b} to terminals $i$ and $j$, reassigning
$C_j=C_i\cup C_j$, and the claim fails to hold for terminal
$j$. Specifically, suppose that for some $j'\in S_a$, after the
iteration we have $r_j,r_{j'}\in C_j\cap C_{j'}$. Then, $r_j\in
C_{j'}$, and therefore $C_{j'}$ intersected $C_j$ prior to the
iteration, and by the induction hypothesis $r_{j'}\in C_i\setminus
C_j$ prior to the iteration. If $r_i\in C_{j'}$, then prior to the
iteration, $i$ and $j'$ contradicted the induction
hypothesis. Otherwise, $i$, $j$ and $j'$ satisfy the conditions in
Step~\ref{step:lam-a'} of the algorithm, and this contradicts the fact
that we apply the transformation in Step~\ref{step:lam-b} at this
iteration.

Next suppose that during some iteration we apply the transformation in
the first part of Step~\ref{step:lam-c} to terminals $i$ and $j$,
reassigning $C_j=C_i\cup C_j$, and the claim fails to hold for
terminal $j$; in particular, for some $j'\in S_a$, after the iteration
we have $r_j,r_{j'}\in C_j\cap C_{j'}$. Then, since $r_j\in C_{j'}$
and the pair of terminals did not match the criteria in
Step~\ref{step:lam-b}, it must be the case that $C_j\subset C_{j'}$
prior to the iteration. Furthermore, $r_{j'}\in C_i$ prior to the
iteration and this contradicts the fact that we applied the
transformation in the first part of Step~\ref{step:lam-c}.

Finally, suppose that during some iteration we apply the
transformation in the second part of Step~\ref{step:lam-c}. Then the
cut assigned to every $i_{x'}$ for $x'\le x-2$ is a subset of the
previous cut of $i_{x'+1}$, but does not contain the latter terminal,
and so by the arguments presented for the previous cases, once again
the induction hypothesis continues to hold for those
terminals. Furthermore, the cut assigned to $i_x$ is a subset of its
original cut and $i_{x-1}$ does not belong to any of the new cuts
except its own. The same argument holds for the $j_{y'}$ terminals.
\end{proof}

\begin{lemma}
\label{lem:lam2-load}
For the cut collection produced by algorithm {\em Integer-Lam-2} the
load on every edge is no more than twice the load of the integral
family of cuts input to the algorithm.
\end{lemma}
\begin{proof}
We first claim that edge loads are preserved throughout
Steps~\ref{step:lam-iter} and \ref{step:lam-d-1} of the
algorithm. This can be established via a case-by-case analysis by
noting that in every transformation of these steps, the number of new
cuts that an edge crosses is no more than the number of old cuts that
the edge crosses prior to the transformation. It remains to analyze
Step~\ref{step:lam-d-2} of the algorithm. We claim that we only lose a
factor of $2$ in edge loads during this step of the algorithm. This is
easy to see.  Note that for every edge $e$, $\sum_{T} p^T_e\le
2\ld{\C_{\cup T}}$, where $\C_{\cup T}$ is the family of cuts
belonging to terminals in any non-singleton component of $\G$ prior to
Step~\ref{step:lam-d-2}. Moreover, in each iteration of the step, we
only load an edge $e$ to the extent of $p^T_e$. Therefore the lemma
follows.
\end{proof}

\begin{proofof}{Lemma~\ref{lem:int-lam2}}
The proof follows immediately from Lemmas~\ref{lem:lam2-laminarity},
\ref{lem:lam2-separation} and \ref{lem:lam2-load}.
\end{proofof}

Given this lemma, algorithm {\em Lam-2} in Figure~\ref{fig:Lam2}
converts an arbitrary feasible solution for \ref{eqn:LP} into a
feasible fractional laminar family.\\

\begin{figure*}[t]
  \begin{small} 
  \rule[0.025in]{\textwidth}{0.01in}
  {\bf Input:} Graph $G=(V,E)$ with edge capacities $c_e$, commodities
  $S_1,\cdots,S_k$, a feasible solution $d$ to the
  program~\ref{eqn:LP}.\\
  {\bf Output:} A fractional laminar family of cuts $\C$ that is
  feasible for $G$ with edge capacities $8c_e+o(1)$.\\
  \rule[0.025in]{\textwidth}{0.01in}

  \begin{enumerate}
    \item \label{step:lam2-metric} For every $a\in [k]$ and every
      terminal $i\in S_a$ do the following: Order the vertices in $G$
      in increasing order of their distance under $d_a$ from
      $r_i$. Let this ordering be $v_0=r_i, v_1, \cdots, v_n$. Let
      $\C^1_i$ be the collection of cuts $\{v_0, v_1, \cdots, v_b\}$,
      one for each $b\in [n]$ with $d_a(r_i,v_b)< 0.5$, with weights
      $w^1(\{v_0, \cdots, v_b\}) =
      2(\min\{d_a(r_i,v_{b+1}),0.5\}-d_a(r_i,v_b))$. Let $\C^1$ denote
      the collection $\{\C^1_i\}_{i\in \cup_a S_a}$.
    \item Let $N=n\sum_a |S_a|$. Round up the weights of all the cuts
      in $\C^1$ to multiples of $1/N^2$, and truncate the collection
      so that the total weight of every sub-collection $\C^1_i$ is
      exactly $1$. Furthermore, split every cut with weight more than
      $1/N^2$ into multiple cuts of weight exactly $1/N^2$ each,
      assigned to the same commodity. Call this new collection $\C^2$
      with weight function $w^2$. Note that every cut in this
      collection has weight exactly $1/N^2$.
    \item Construct a new instance of MCP in the same graph $G$ as
      follows. For each $a\in [k]$, construct $N^2$ new commodities
      with terminal sets identical to that of $S_a$ (that is the
      terminals reside at the same nodes). For every new terminal
      corresponding to an older terminal $i$, assign to the new
      terminal a unique cut from $\C^2_i$ with weight $1$. Call this
      new collection $\C^3$, and the new instance $I$.
    \item Apply algorithm {\em Integer-Lam-2} from
      Figure~\ref{fig:Int-Lam2} to the family $\C^3$ to obtain family
      $\C^4$.
    \item For every $a\in [k]$ and every $i\in S_a$, let $\C^5_i$ be
      the set of $N^2/2$ innermost cuts in $\C^4$ assigned to
      terminals in the new instance $I$ that correspond to terminal
      $i$. (Note that these cuts are concentric as they belong to a
      laminar family and all contain $r_i$. Therefore ``innermost''
      cuts are well defined.) Assign a weight of $2/N^2$ to every cut
      in this set. Output the collection $\C^5$.
  \end{enumerate}

  \rule[0.025in]{\textwidth}{0.01in}
  \caption{Algorithm {\em Lam-2}---Algorithm to convert an LP solution
  into a feasible fractional laminar family}
  \label{fig:Lam2}
  \end{small}
\end{figure*}

\setcounter{theorem}{2}
\addtocounter{theorem}{-1}
\begin{lemma}
Consider an instance of the MCP with graph $G=(V,E)$, edge capacities
$c_e$, and commodities $S_1,\cdots,S_k$. Given a feasible solution $d$
to \ref{eqn:LP}, algorithm {\em Lam-2} produces a fractional laminar
cut family $\C$ that is feasible for the MCP on $G$ with edge
capacities $8c_e+o(1)$.
\end{lemma}
\setcounter{theorem}{18}

\begin{proof}
Note first that the cut collection $\C^1$ satisfies the following
properties: (1) For every $a\in [k]$ and $i\in S_a$, every cut in
$\C^1_i$ contains $r_i$, but not $r_j$ for $j\in S_a$, $j\ne i$; (2)
The total weight of cuts in $C^1_i$ is $1$; (3) For every edge $e$,
$\ld{\C^1} \le 2\sum_{a} d_a(e)\le 2c_e$. The family $\C^2$ also
satisfies the first two properties, however loads the edges slightly
more than $\C^1$. Any edge belongs to at most $N$ cuts, and therefore
the load on the edge goes up by an additive amount of at most
$1/N$. Therefore, for every $e$, $\ld{\C^2}\le 2c_e+1/N$. Next, the
collection $\C^3$ is a feasible integral family of cuts for the new
instance $I$ with $\ld{\C^3}=N^2\ld{\C^2}$. Therefore, applying
Lemma~\ref{lem:int-lam2}, we get that $\C^4$ is a feasible laminar
integral family of cuts for $I$ with $\ld{\C^4}\le
2N^2(2c_e+1/N)$. Finally, in family $\C^5$, every terminal $i\in S_a$
gets assigned $N^2/2$ fractional cuts, each with weight
$2/N^2$. Therefore, the total weight of cuts in $\C^5_i$ is $1$. Now
consider any two terminals $i,j\in S_a$ with $i\ne j$. Then, in all
the $N^2$ commodities corresponding to $S_a$ in instance $I$, either
the cut assigned to $i$'s counterpart, or that assigned to $j$'s
counterpart separates $i$ from $j$. Say that among at least $N^2/2$ of
the commodities in $I'$, the cut assigned to $i$'s counterpart
separates $i$ from $j$. Then, the innermost $N^2/2$ cuts assigned to
$i$ in $\C^5$ separate $i$ from $j$. Therefore, the family $\C^5$
satisfies the first two conditions of feasibility as given in
Definition~\ref{def:feas}. Finally, it is easy to see that on every
edge $e$, $\ld{\C^5}\le 2/N^2\ld{\C^4}\le 4(2c_e+1/N)$.
\end{proof}

\section{NP-Hardness}
\label{sec:NP-hard}

We will now prove that CSCP and MCP are NP-hard. Since edge loads for
any feasible solution to these problems are integral, the result of
Theorem~\ref{thm:main} is optimal for the CSCP assuming P$\ne$NP. The
reduction in this theorem also gives us an integrality gap instance
for the CSCP.

\begin{theorem}
CSCP and MCP are NP-hard. Furthermore the integrality gap of
\ref{eqn:LP} is at least $2$ for both the problems. 
\end{theorem}
\begin{proof}
We reduce independent set to CSCP. In particular, given a graph $G$
and a target $k$, we produce an instance of CSCP such that the load on
every edge is at most $1$ if and only if $G$ contains an independent
set of size at least $k$. Let $n$ be the number of vertices in $G$. We
construct $G'$ by adding a chain of $n-k+1$ new vertices to $G$. Let
the first vertex in this chain be $t$ (the common sink) and the last
be $v$. We connect every vertex of $G$ to the new vertex $v$, and
place a terminal $i$ at every vertex $r_i$ in $G$ (therefore, there
are a total of $n$ sources). We claim that there is a collection of
$n$ edge-disjoint $r_i-t$ cuts in this new graph $G'$ if and only if
$G$ contains an independent set of size $k$.

One direction of the proof is straightforward: if $G$ contains an
independent set of size $k$, say $S$, then for each vertex $r_i\in S$,
consider the cut $\{r_i\}$, and for each of the $n-k$ source not in
$S$, consider the cuts obtained by removing one of the $n-k$ chain
edges in $G'$. Then all of these $n$ cuts are edge-disjoint.

Next suppose that $G'$ contains a collection of edge-disjoint cuts
$C_i$, with $r_i\in C_i$ and $t\not\in C_i$ for all $i$. Note that the
number of cuts $C_i$ containing any chain vertex is at most $n-k$
because each of them cuts at least one chain edge. Next consider the
cuts that do not contain any chain vertex, specifically $v$, and let
$T'$ be the collection of terminals for such cuts. These are at least
$k$ in number. Note that any cut $C_i$, $i\in T'$, cuts the edges
$(u,v)$ for $u\in C_i$. Therefore, in order for these cuts to be
edge-disjoint, it must be the case that $C_i\cap C_j = \emptyset$ for
$i,j\in T'$, $i\ne j$. Finally, for two such cuts $C_i$ and $C_j$,
edge-disjointness again implies that $r_i$ and $r_j$ are not
connected. Therefore the vertices $r_i$ for $i\in T'$ form an
independent set in $G$ of size at least $k$.

For the integrality gap, let $G$ be the complete graph and $k$ be
$n/2$. Then, there is no integral solution with load $1$ in
$G'$. However, the following fractional solution is feasible and has a
load of $1$: let the chain of vertices added to $G$ be $v=v_1, v_2,
\cdots, v_{n/2+1}=t$; assign to every terminal $i$, $i\in [n]$, the
cut $\{r_i\}$ with weight $1/2$, and the cut $V\cup
\{v_0,\cdots,v_{\lfloor i/2\rfloor}\}$ with weight $1/2$.
\end{proof}

\section{Concluding Remarks}
\label{sec:lam-gap}

Given that our algorithms rely heavily on the existence of good
laminar solutions, a natural question is whether every feasible
solution to the MCP can be converted into a laminar one with the same
load. Figure~\ref{fig:lam-gap} shows that this is not true. The figure
displays one integral solution to the MCP where the solid edges
represent the cut for commodity $a$, and the dotted edges represent
the cut for commodity $b$. It is easy to see that this instance admits
no fractional laminar solution with load $1$ on every edge.

Is the ``laminarity gap'' small for the more general set multiway cut
packing and multicut packing problems as well? We believe that this is
not the case and there exist instances for both of those problems with
a non-constant laminarity gap.

\begin{figure}[]
\begin{center}
\epsfig{file = 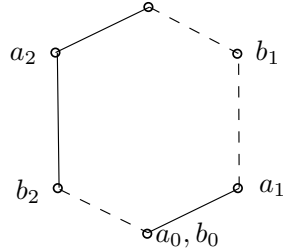, height = 1.3in}
\caption{Each edge has capacity $1$. There are two commodities with
terminal sets $\{a_0, a_1, a_2\}$ and $\{b_0, b_1, b_2\}$.}
\label{fig:lam-gap}
\end{center}
\end{figure}



\begin{small}
\bibliographystyle{plain}
\bibliography{cuts}
\end{small}

\end{document}